\documentclass{amsart}
\usepackage[foot]{amsaddr}
\usepackage[T1]{fontenc}

\usepackage{amssymb}
\usepackage{mathtools}
\usepackage{enumitem}
\usepackage{amsmath}
\usepackage{amsfonts}
\usepackage{amscd}
\usepackage{amsthm}
\usepackage{graphicx}
\usepackage{hyperref}

\newtheorem{theorem}{Theorem}

\newtheorem{corollary}{Corollary}

\newtheorem{lemma}{Lemma}

\theoremstyle{remark}

\theoremstyle{definition}

\title[Independence Covering in \texorpdfstring{$C_4$}{4-cycle}-Free Graphs]{Lower Bound for Independence Covering in \texorpdfstring{$C_4$}{4-cycle}-Free Graphs}
\author{Michael Kuhn$^{*}$}
\email{mkuhn@ucsb.edu}
\thanks{$^{*}$Thanks to the donors to the Science/Mathematics Summer Undergraduate Research Fellowship Fund at CCS for supporting my 2022 Summer Undergraduate Research Fellowship}
\author{Daniel Lokshtanov}
\email{daniello@ucsb.edu}
\author{Zachary Miller$^{\dagger}$}
\email{zmiller@ucsb.edu}
\thanks{$^{\dagger}$Thanks to the donors to the Computer Science Endowment at CCS for supporting my 2022 Summer Undergraduate Research Fellowship}
\address{University of California Santa Barbara, Santa Barbara CA 93106}
\date{August 29, 2023}

\begin{document}
    \begin{abstract}
        An \emph{independent set} in a graph $G$ is a set $S$ of pairwise non-adjacent vertices in $G$. A family $\mathcal{F}$ of independent sets in $G$ is called a $k$-\emph{independence covering family} if for every independent set $I$ in $G$ of size at most $k$, there exists an $S \in \mathcal{F}$ such that $I \subseteq S$.
        Lokshtanov et al. [ACM Transactions on Algorithms, 2018] showed that graphs of degeneracy $d$ admit $k$-independence covering families of size $\binom{k(d+1)}{k} \cdot 2^{o(kd)} \cdot \log n$, and used this result to design efficient parameterized algorithms for a number of problems, including \textsc{Stable Odd Cycle Transversal} and \textsc{Stable Multicut}. 

        In light of the results of Lokshtanov et al. it is quite natural to ask whether even more general families of graphs admit $k$-independence covering families of size $f(k)n^{O(1)}$. 
        Graphs that exclude a complete bipartite graph $K_{d+1,d+1}$ with $d+1$ vertices on both sides as a subgraph, called $K_{d+1,d+1}$-\emph{free graphs}, are a frequently considered generalization of $d$-degenerate graphs. 
%
        This motivates the question whether $K_{d,d}$-free graphs admit $k$-independence covering families of size $f(k,d)n^{O(1)}$. Our main result is a resounding ``no'' to this question -- specifically we prove that even $K_{2,2}$-free graphs (or equivalently $C_4$-free graphs) do not admit $k$-independence covering families of size  $f(k)n^{\frac{k}{4}-\epsilon}$. 
    \end{abstract}
    \maketitle
\newpage
    \section{Introduction}
An \emph{independent set} in a graph $G$ is a set $S$ of vertices in $G$ such that no two distinct vertices in $S$ are adjacent in $G$.
A family $\mathcal{F}$ of independent sets in $G$ is said to be a $k$-\emph{independence covering family} if, for every independent set $I$ in $G$ of size at most $k$, there exists an $S \in \mathcal{F}$ such that $I$ is a subset of $S$.
Every $n$-vertex graph $G$ has a $k$-independence covering family of size at most $n^k$, namely the family of all independent sets of size at most $k$.
Lokshtanov et al.~\cite{ind-cov-fam} observed that many graphs have substantially smaller $k$-independence covering families, and that $k$-independence covering families of sufficiently small size (in particular of size $f(k)n^{O(1)}$) are very useful for designing parameterized algorithms for certain problems, including \textsc{Stable $s$-$t$-Separator}, \textsc{Stable Odd Cycle Transversal} and \textsc{Stable Multicut} (we refer to Loksthtanov et al.~\cite{ind-cov-fam} for the definitions of these problems).

Lokshtanov et al. gave constructions of $k$-independence covering families of size $\binom{k(d+1)}{k} \cdot 2^{o(kd)} \cdot \log n$ for every $d$-degenerate graph (a graph $G$ is $d$-\emph{degenerate} if every induced subgraph of $G$, including $G$ itself, has a vertex of degree at most $d$), and of size $f(k)n^{\epsilon}$ for every $\epsilon > 0$ and every nowhere-dense family of graphs (the definition of nowhere-dense families of graphs is not relevant for this article, so we omit it).
These constructions have later been used to design parameterized algorithms for a number of different graph problems~\cite{AgrawalHM22,AgrawalJKS20,FellowsGMMRRSS23,JacobMR21,JainKM20}. 

A limiting factor for the applicability of $k$-independence covering families is that the constructions of Lokshtanov et al.~\cite{ind-cov-fam} do not work for all graphs, but only for $d$-degenerate and nowhere dense families of graphs. 
One might ask whether this limitation is necessary; could it be that every $n$-vertex graph $G$ has a $k$-independence covering family of size $f(k)n^{O(1)}$?
This question has a simple negative answer -- in the disjoint union of $k$ complete graphs on $n/k$ vertices, any $k$-independence covering family must have size at least $(n/k)^k$.

This motivates the research question of this paper: what is the most general family of graphs in which every graph has a $k$-independence covering family of size $f(k)n^{O(1)}$?
A tempting target is the class of $K_{d,d}$-\emph{free} graphs. For every integer $d \geq 1$, $K_{d,d}$ is the complete bipartite graph with $d$ vertices on both sides (i.e, with vertices $v_1, \ldots, v_{2d}$ and edge set $\{v_iv_j ~:~ i \leq d < j\}$). A graph $G$ is $K_{d,d}$-\emph{free} if one cannot obtain $K_{d,d}$ from $G$ by deleting vertices and edges. 

The reason why studying $k$-independence covering of $K_{d,d}$-free graphs is particularly appealing is because both $d$-degenerate graphs and nowhere dense families of graphs exclude a $K_{d',d'}$ for some $d'$. Therefore a $k$-independence covering family of size $f(k)n^{O(1)}$ for $K_{d,d}$-free graphs would yield a common generalization of the two constructions of Lokshtanov et al.~\cite{ind-cov-fam} for $d$-degenerate and nowhere dense families of graphs. 
Our main result is that, unfortunately, for every function $f$ there exist $K_{d,d}$-free graphs that do not admit $k$-independence covering families of size at most $f(k)n^{O(1)}$; this holds for any choice of $d \ge 2$. In fact, we show that a well known example of bipartite $K_{2,2}$-free graphs (the $K_{2,2}$ is usually referred to as $C_4$, the cycle on $4$ vertices), namely the point-line incidence graphs of projective planes, do not admit $k$-independence covering families of size at most  $f(k)n^{\frac{k}{4}-\epsilon}$. While the example graph is well known, the proof of the lower bound requires some effort and the application of a fairly recent isoperimetric inequality for such graphs~\cite{vertex-iso}.
On the positive side we show that $C_4$-free graphs do admit $k$-independence covering families of size at most  $k^{k+O(1)} \cdot n^{k/2 + o(k)}$. This bound easily follows from the fact that $C_4$-free graphs are $\lceil\sqrt{n}\rceil$-degenerate together with the construction of $k$-independence covering families of Lokshtanov et al~\cite{ind-cov-fam}.
    
\section{Lower Bound for Independence Covering in \texorpdfstring{$C_4$}{4-cycle}-free graphs}\label{nex}    
All graphs considered in this work are simple and undirected. We denote by $V(G)$ and $E(G)$ the set of vertices and edges of $G$, respectively. The \emph{neighborhood} of a vertex $v$ in a graph $G$ is the set $N(v) = \{u \in V(G) ~:~ uv \in E(G)\}$. For a vertex set $S$ the set $N(S)$ is defined as $N(S) = \left(\bigcup_{u \in S} N(u)\right) - S$. A graph $G$ is said to be \emph{bipartite} if there exists a partition $V_1$, $V_2$ of $V(G)$ such that every edge of $G$ has one endpoint in $V_1$ and the other in $V_2$. The sets $V_1$ and $V_2$ are called the \emph{sides} or \emph{bipartitions} of $G$.

For every prime $q$ we construct a graph $\Gamma_{q}$. The graph can be succintly described as the Levi graph (i.e point-line incidence graph) of a projective plane of order $q$ (see Section 3 of~\cite{proj-plane} for a construction of projective planes of order $q$). 
In order to keep the presentation self contained we give a full description of the graph here, and neither the construction of $\Gamma_{q}$ nor our lower bound on the size of independence covering families require the reader to know projective planes. 

\smallskip
\noindent
{\bf The graph $\Gamma_q$.} For a prime $q$, let $\mathbb{Z}_q$ be the finite field of order $q$, i.e. the set of integers $\{0, \ldots, q-1\}$ equipped with the $+$ and $\cdot$ operation modulo $q$. The set $P$ has the following $q^2 + q + 1$ vertices:
\emph{(i)} a vertex $p_{x,y}$ for every  pair $(x, y) \in \mathbb{Z}_q \times \mathbb{Z}_q$,
\emph{(ii)} a vertex $\hat{p}_a$ for every element $a \in \mathbb{Z}_q$, and
\emph{(iii)} a vertex $\tilde{p}$.
The set $L$ has the following $q^2 + q + 1$ vertices:
\emph{(i)} a vertex $\ell_{a,b}$ for every  pair $(a, b) \in \mathbb{Z}_q \times \mathbb{Z}_q$,
\emph{(ii)} a vertex $\hat{\ell}_x$ for every element $x \in \mathbb{Z}_q$, and
\emph{(iii)} a vertex $\tilde{\ell}$.
We describe the set of edges of $\Gamma_{q}$ by describing for every vertex of $P$, the set of its adjacent vertices in $L$.
\emph{(i)} For each $(x, y) \in \mathbb{Z}_q \times \mathbb{Z}_q$ the vertex $p_{x,y}$ is adjacent to all vertices $\ell_{a,b}$ such that $ax + b \equiv y \mod q$, as well as to $\hat{\ell}_x$.
\emph{(ii)} For each $a \in \mathbb{Z}_q$ the vertex $\hat{p}_a$ is adjacent to all vertices $\ell_{a,b}$ (for all $b \in \mathbb{Z}_q$), as well as to $\tilde{\ell}$.
\emph{(iii)} The vertex $\tilde{p}$ is adjacent to $\hat{\ell}_x$ for every element $x \in \mathbb{Z}_q$, as well as to $\tilde{\ell}$. This concludes the construction of $\Gamma_q$.

A way to visualize the graph $\Gamma_q$ is to draw the vertices $p_{x,y}$ in the Cartesian plane in the coordinates $(x,y)$. Each vertex $\ell_{a,b}$ corresponds to a line with slope $a$ which passes through the $y$-axis in the point $(0, b)$ and ``wraps around'' modulo $q$.
The vertex $p_{x,y}$ is adjacent to $\ell_{a,b}$ if the point lies on the corresponding line.
For each $a \in \mathbb{Z}_q$ the vertex  $\hat{p}_a$ corresponds to the ``point out in infinity'' where all the lines with slope $a$ meet. 
For each $x \in \mathbb{Z}_q$ the vertex  $\hat{\ell}_x$ corresponds to the vertical line that pass through the $x$-axis in $(x, 0)$.
The vertex $\tilde{p}$ is the ``point out in infinity'' where all the vertical lines  $\hat{\ell}_x$ meet, and $\tilde{\ell}$ is the ``line'' that passes through all points out in infinity.  The visual for $\Gamma_2$ is shown in \autoref{fig:fano}.
\begin{figure}
    \centering
    \includegraphics[width=0.5\textheight]{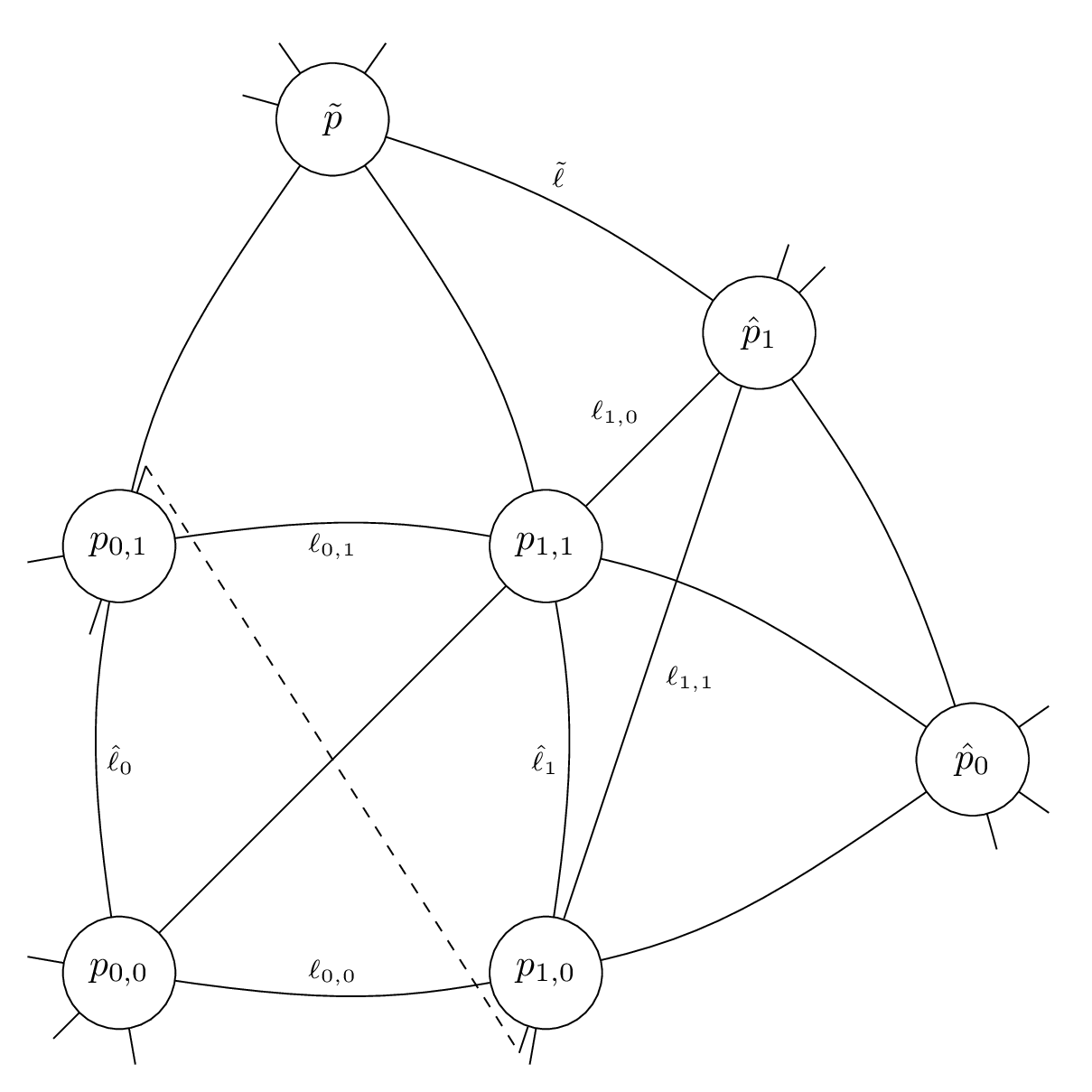}
    \caption{The \textit{fano plane}, $\Gamma_2$.}
    \label{fig:fano}
\end{figure}

We summarize without proof a few properties of the graph $\Gamma_q$. These properties can be verified directly from the definition of $\Gamma_q$.
\begin{lemma}\label{lem:properties}
For every prime $q$ the graph $\Gamma_q$ satisfies the following properties. 
\begin{enumerate}
\item $\Gamma_q$ has $n = 2(q^2 + q + 1)$ vertices. 
\item For every $p \in P$, $|N(p)| = q + 1$.
\item\label{itm:oneCommon} For every pair of distinct vertices $p_1, p_2$ in $P$, we have $|N(p_1) \cap N(p_2)| = 1$.
\item For every $\ell \in L$, $|N(\ell)| = q + 1$.
\item For every pair of distinct vertices $\ell_1, \ell_2$ in $L$, we have $|N(\ell_1) \cap N(\ell_2)| = 1$.
\end{enumerate}
\end{lemma}

In the rest of this section we will prove that for every function $f : \mathbb{N} \rightarrow \mathbb{N}$, integer $k$ and $\epsilon > 0$, when $q$ (and therefore $n$) is sufficiently large compared to $k$ then $\Gamma_q$ does not have a $k$-independence covering family of size  $f(k)n^{\frac{k}{4}-\epsilon}$.  The fact that $\Gamma_q$ is bipartite with sides $P$ and $L$, as well as the properties listed in \autoref{lem:properties} are the only properties of $\Gamma_q$ that we will use in this proof. The fact that $\Gamma_q$ is $C_4$-free follows immediately from Property~\ref{itm:oneCommon} of \autoref{lem:properties}.

We say that a graph $G$ is a \textit{$(\eta,\Delta,\lambda)$-graph} if $G$ is a bipartite graph with bipartition $(V_1,V_2)$ with $|V_1|=|V_2| = \eta$, every vertex in $G$ has degree exactly $\Delta$, and for every $i \in \{1,2\}$ and every pair of distinct vertices in $V_i(G)$ have exactly $\lambda$ common neighbors. \autoref{lem:properties} states precisely that $\Gamma_q$ is a $(q^2+q+1,q + 1,1)$-graph.
A key ingredient of our proof is the following Lemma by Price et al.~\cite{vertex-iso}, which states that in $(\eta,\Delta,\lambda)$-graphs where $\lambda$ is small compared to $\Delta$, every set $S$ of size up to about $\Delta/\lambda$ has a large neighborhood (relative to the size of $S$). 

\begin{lemma}[\cite{vertex-iso}]\label{lem:vklam}
            For every $(\eta,\Delta,\lambda)$-graph $G$,  and a nonempty subset $S$ of $V_1(G)$ or $V_2(G)$ it holds that
            \begin{equation*} \label{lower0}
                \frac{|N(S)|}{|S|} \ge \frac{\Delta^2}{\Delta + \lambda (|S| - 1)}\mbox{.}
            \end{equation*}
\end{lemma}

        Applying \autoref{lem:vklam} to $\Gamma_q$ immediately yields the following corollary.
        
        \begin{corollary}\label{cor:expansion}
            In $\Gamma_{q}$, every nonempty subset $S$ of $P$ or $L$ satisfies
            \begin{equation}\label{lower2}
                |N(S)| \ge \frac{(q+1)^2|S|}{q + |S|}.
            \end{equation}
        \end{corollary}

        \begin{proof}
            From \autoref{lem:vklam} we have that
            \begin{align*}\label{lower1}
                \frac{|N(S)|}{|S|} \ge \frac{(q+1)^2}{(q+1) + (|S| - 1)} = \frac{(q+1)^2}{q + |S|}.
            \end{align*}
            Multiplying on both sides by $|S|$ proves the statement of the corollary.
        \end{proof}

        The main idea of our construction is to use \autoref{cor:expansion} in order to show that no independent set in $\Gamma_q$ can have many vertices both in $P$ and in $L$. This in turn shows that no independent set in $\Gamma_q$ can contain many independent subsets with $k/2$ vertices in $P$ and $k/2$ vertices in $L$. For sufficiently small $k$, a ``large fraction'' of sets of $k/2$ vertices in $P$ and $k/2$ vertices in $L$ form an independent set, and a simple counting argument then yields the lower bound on the size of $k$-independence covering families. We start by showing that no independent set in $\Gamma_q$ can have many vertices both in $P$ and in $L$.

        \begin{lemma}\label{lem:productBound}
        Let $I$ be an independent set in $\Gamma_q$, $a = |I \cap P|$ and $b = |I \cap L|$. Then $ab \leq q(q+1)^2 < 2n^{3/2}$.
        \end{lemma}

        \begin{proof}
        Let $S = I \cap P$ and $T = I \cap L$. Then $a = |S|$, and we have that 
                \begin{align*}
                    b &\le |V_2(\Gamma_q)| - |N(S)| & \mbox{since } I \mbox{ is independent}\\
                    &\le ((q + 1)^2 - q) - \frac{(q + 1)^2a}{q + a} & \mbox{by \autoref{cor:expansion}}\\
                    &= \frac{q(q + 1)^2 - q(q + a)}{q + a} \\
                    &\le \frac{q(q + 1)^2}{a}.
                \end{align*}
        Multiplying both sides by $a$ yields $ab \leq q(q+1)^2 = q(n + q)$. Now $q < \sqrt{n} < n$, hence the bound can be simplified to $q(n+q) < \sqrt{n}(n + q) < 2n^{3/2}$, proving the statement of the lemma.
        \end{proof}

        \begin{corollary}\label{cor:coverBound}
        Let $I$ be an independent set in $\Gamma_q$ and $k$ be an even integer. Then there are at most $2^{k/2}n^{3k/4}$ distinct vertex sets $Z \subseteq I$ such that $|Z \cap P| = |Z \cap L| = k/2$.
        \end{corollary}

        \begin{proof}
        Let $a = |I \cap P|$ and $b = |I \cap L|$. If $\min(a,b) < k/2$ there are no choices for $Z$ and so the statement of the corollary follows. If $\min(a,b) \geq k/2$ there are 
        $$\binom{a}{k/2}\binom{b}{k/2} \leq (ab)^{\frac{k}{2}} < (2n^{3/2})^{\frac{k}{2}}$$
        choices for $Z$, and the statement follows. Here the last inequality follows from \autoref{lem:productBound}.
        \end{proof}

        \autoref{cor:coverBound} implies that in a $k$-independence covering family $\mathcal{F}$ of $\Gamma_q$ each set $I \in \mathcal{F}$ contains no more than $2^{k/2}n^{3k/4}$ independent sets that have $k/2$ vertices both in $P$ and in $L$. We now show that the number of such independent sets in $\Gamma_q$ is much larger, in particular quite close to $n^k$. This immediately implies a lower bound on $|\mathcal{F}|$ of about $n^{k/4}$.

        \begin{lemma}\label{lem:ifp}
            Let $k \leq q$ be an even integer. The number of independent sets $Z$ in $\Gamma_{q}$ such that $|Z \cap P| = |Z \cap L| = k/2$ is at least $\left(\frac{n}{4k}\right)^k$.
        \end{lemma}
        \begin{proof}
            There are $\binom{|P|}{k/2} = \binom{n/2}{k/2}$ choices for $Z \cap P$. Since every vertex in $\Gamma_q$ has degree $q+1$ and $k \leq q$ we have that every choice of $Z \cap P$ satisfies
            $$|N(Z \cap P)| \leq (q+1)\frac{k}{2} \leq \frac{q^2+q}{2} < \frac{n}{4}\mbox{.}$$
            Hence, for each choice of $Z \cap P$ there are at least $\binom{n/4}{k/2}$ choices for $Z \cap L$. Therefore there are at least $\binom{n/2}{k/2}\binom{n/4}{k/2}$ distinct choices for $Z$. Now,
            
            $$\binom{n/2}{k/2}\binom{n/4}{k/2} > \binom{n/4}{k/2}^2 > \left(\frac{(n/4 - k/2)^{\frac{k}{2}}}{(k/2)^{\frac{k}{2}}}\right)^2 = \left(\frac{n/2 - k}{k}\right)^k > \left(\frac{n}{4k}\right)^k\mbox{.}$$
            Here the last transition holds because $k \leq q < n/4$.

        \end{proof}
        
        We are now ready to lower bound the size of $k$-independence covering families of $\Gamma_q$.

        \begin{theorem}\label{thm}
            Let $k \leq q$ be an even integer and $\mathcal{F}$ be a $k$-independence covering family of $\Gamma_q$. Then $|\mathcal{F}| \geq \frac{n^{k/4}}{(4\sqrt{2}k)^k}$.
        \end{theorem}

        \begin{proof}
        By \autoref{lem:ifp} there are at least $\left(\frac{n}{4k}\right)^k$ independent sets $Z$ of size $k$ in $\Gamma_q$ such that $|Z \cap P| = |Z \cap L| = k/2$. For each such $Z$ there exists an $I \in \mathcal{F}$ such that $Z \subseteq I$. However, by \autoref{cor:coverBound} there are no more than $2^{k/2}n^{3k/4}$ distinct vertex sets $Z \subseteq I$ such that $|Z \cap P| = |Z \cap L| = k/2$. Thus, 
        $$|\mathcal{F}| \geq \frac{\left(\frac{n}{4k}\right)^k}{2^{k/2}n^{3k/4}} = \frac{n^{k/4}}{(4\sqrt{2}k)^k}\mbox{.}$$
        \end{proof}
      
        Since $k$ can be chosen arbitrarily small compared to $q$ (and therefore $n$) in the statement of \autoref{thm}, this rules out the possibility of $k$-independence covering families of size $f(k)n^{\frac{k}{4}-\epsilon}$ for $C_4$-free graphs. Indeed, for a given function $f: \mathbb{N} \to \mathbb{N}$, integer $k$ and $\epsilon > 0$ we choose a prime $q$ such that $k < q$ and $n^{\epsilon} > f(k)(4\sqrt{2}k)^k$. Then \autoref{thm} yields that the size of any $k$-independence covering family for $\Gamma_q$ must be at least $\frac{n^{k/4}}{(4\sqrt{2}k)^k}$, which is strictly more than $f(k)n^{\frac{k}{4}-\epsilon}$.

\section{Upper Bound for Independence Covering on \texorpdfstring{$C_4$}{4-cycle}-free Graphs}\label{ex}
    We now show that $C_4$-free graphs admit an independence-covering lemma which is slightly better than naively iterating over all independent sets of size $k$. Our result follows directly by combining the independence-covering lemma for $d$-degenerate graphs of Loksthanov et al.~\cite{ind-cov-fam} with a classic degeneracy bound for $C_4$-free graphs. We start with the degeneracy bound. The proof closely follows an upper bound on the number of edges in $d$-degenerate graphs attributed to Rieman (see \cite{virt-edge}).
    
    \begin{lemma}\label{lem:degen} Every $C_4$-free graph $G$ is $\lceil\sqrt{n}\rceil$-degenerate. 
    \end{lemma}
    
    \begin{proof}
    Since every subgraph of a $C_4$-free graph is $C_4$-free it suffices to show that every $C_4$-free graph has a vertex of degree at most $\lceil\sqrt{n}\rceil$. Suppose not, and define for every vertex $v$ in $G$ the set $E_v$ to be the set of unordered pairs $xy$ of distinct vertices in $N(v)$. For every $v\in V(G)$ we have $|E_v| = \binom{|N(v)|}{2} \geq n/2$.
    Since the number of distinct unordered pairs of vertices in $G$ is $n(n-1)/2 < n^2/2 \leq \sum_{v \in V(G)} |E_v|$ the pigeon hole principle yields that there exists a pair $xy$ and vertices $u$, $v$ such that $xy \in E_u \cap E_v$. But then $u,x,v,y$ is a cycle on $4$ vertices in $G$, contradicting that $G$ is $C_4$-free.
    \end{proof}

    We now re-state Lemma 3.2 of Lokshtanov et al.~\cite{ind-cov-fam}. The upper bounds stated here are slightly less sharp than they are in~\cite{ind-cov-fam} because we will only need to apply the lemma with $d = \lceil\sqrt{n}\rceil$.
    
    \begin{lemma}[Lemma 3.2 of~\cite{ind-cov-fam}]\label{lem:degenCover} 
    There is an algorithm that given a $d$-degenerate $n$-vertex graph $G$ and $k \in \mathbb{N}$, runs in time $\binom{k(d+1)}{k}^{1+o(1)}n^{O(1)}$ and outputs a $k$-independence covering family for $G$ of size at most  $\binom{k(d+1)}{k}^{1+o(1)}k(d+1) \log n$.
    \end{lemma}

    Since $C_4$-free graphs are $\lceil\sqrt{n}\rceil$-degenerate (by \autoref{lem:degen}), \autoref{lem:degenCover} with $d =\lceil\sqrt{n}\rceil$ immediately implies the following independence covering lemma for $C_4$-free graphs.
    
    \begin{lemma}\label{lem:C4freeCover} 
    There is an algorithm that given an $n$-vertex $C_4$-free graph $G$ and $k \in \mathbb{N}$, runs in time $k^{k+O(1)} \cdot n^{k/2 + o(k)}$ and outputs a $k$-independence covering family for $G$ of size at most $k^{k+O(1)} \cdot n^{k/2 + o(k)}$.
    \end{lemma}

\section{Conclusion}
   We showed that $C_4$-free graphs do not admit $k$-independence covering families of size $f(k)n^{k\frac{1-\epsilon}{4}}$ for any function $f$. We remark that this lower bound holds not only against independence covering families, but also against randomized independence covering lemmas on the form of Lemma 1 of Lokshtanov et al.~\cite{ind-cov-fam}. In particular, there cannot exist a probability distribution on the independent sets of $\Gamma_q$ such that for every independent set $Z$ of size at most $k$ the probability that the sampled independent set $I$ contains $Z$ is at least $(f(k)n^{k\frac{1-\epsilon}{4}})^{-1}$. Indeed, sampling $(f(k)n^{k\frac{1-\epsilon}{4}}) \cdot O(k \log n)$ independent sets from such a distribution would produce with non-zero probability a $k$-independence covering family for $\Gamma_q$, contradicting \autoref{thm}.

    Our work leaves a gap between the $n^{k/4}$ lower bound and the $n^{k/2}$ upper bound for the size of $k$-independence covering families of $C_4$-free graphs, and closing this gap might be interesting. Perhaps more interesting is to obtain a more complete understanding of which graph families aside from $d$-degenerate and nowhere dense admit $k$-independence covering families of size $f(k)n^{O(1)}$.

\newpage
\bibliographystyle{plainurl}
\bibliography{refs}

\end{document}